\documentclass{article}

\usepackage[preprint]{neurips_2019}

\usepackage[utf8]{inputenc} 
\usepackage[T1]{fontenc}    
\usepackage{hyperref}       
\usepackage{url}            
\usepackage{booktabs}       
\usepackage{amsfonts}       
\usepackage{nicefrac}       
\usepackage{microtype}      

\usepackage{amsmath,amssymb,amsthm}

\usepackage{algorithm}
\usepackage{algorithmic}
\usepackage{amsmath}

\usepackage{graphicx}
\usepackage{caption}
\usepackage{subfigure} 

\usepackage{comment}

\usepackage{xspace}

\makeatletter
\DeclareRobustCommand\onedot{\futurelet\@let@token\@onedot}
\def\@onedot{\ifx\@let@token.\else.\null\fi\xspace}

\def\ie{\emph{i.e}\onedot}

\makeatother

\usepackage{tikz}
\usetikzlibrary{arrows}

\usepackage{mathtools}
\DeclarePairedDelimiter{\ceil}{\lceil}{\rceil}
\DeclarePairedDelimiter\floor{\lfloor}{\rfloor}

\newtheorem{theorem}{Theorem}

\newtheorem{lemma}{Lemma}

\newtheorem{definition}{Definition}

\graphicspath{{images/}} 

\title{Communication-Efficient Weighted Sampling and Quantile Summary for GBDT}

\author{
  Ziyue Huang, Ke Yi\\
  \texttt{zhuangbq@cse.ust.hk, yike@cse.ust.hk} \\
  Department of Computer Science\\
  Hong Kong University of Science and Technology
}

\begin{document}

\maketitle

\begin{abstract}
Gradient boosting decision tree (GBDT) is a powerful and widely-used machine learning model, which has achieved state-of-the-art performance in many academic areas and production environment. However, communication overhead is the main bottleneck in distributed training which can handle the massive data nowadays. In this paper, we propose two novel communication-efficient methods over distributed dataset to mitigate this problem, a weighted sampling approach by which we can estimate the information gain over a small subset efficiently, and distributed protocols for weighted quantile problem used in approximate tree learning.
\end{abstract}

\section{Introduction}

Gradient boosting decision tree (GBDT) \cite{friedman2001greedy} is a powerful and widely-used machine learning model, which has been highly successful in many applications. Traditional GBDT training requires scanning the whole dataset to find the best split point for each feature with largest information gain, which is the major bottleneck in the training process. To improve the efficiency and scalability, two popular approaches are down sampling the dataset and approximated tree learning using quantile strategy. 

Due to the surge of big data, distributed training has received a lot of attention recently, since it is an efficient way to handle large dataset by leveraging the computational power of multiple machines. However, its scalability is limited by the possibly overwhelming communication cost. In this paper, we propose two novel communication-efficient methods over distributed dataset to mitigate this problem.

\begin{itemize}
    \item Weighted sampling, by which we can sample a subset, where each data instance is included with probability proportional to the magnitude of its gradient, to estimate information gain of a given candidate split point, instead of scanning the whole dataset.
    \item A one-round protocol and a tree shape all-reduce protocol for weighted quantile problem used in approximate tree learning.
\end{itemize}

\section{Related Work}

To reduce the number of data instances used in the training process, GOSS \cite{ke2017lightgbm} estimates the information gain of each candidate split point over a small subset, which is composed of a small portion of data instances with larger gradients and a uniform sample from the rest data instances.

For approximate tree learning, Chen {\em et al.\/} \cite {chen2016xgboost} choose the best split point over a small candidate set, where the candidate split points are evenly distributed on the entire dataset according to the second order gradient statistics. They proposed a deterministic $\varepsilon$-approximate weighted quantile summary, which extends GK summary \cite{greenwald2001space} with merge and prune operations on weighted items as a basic building block and follows the multi-level summary structure from \cite{zhang2007fast}. Ivkin {\em et al.\/} proposed another randomized weighted quantile algorithm \cite{ivkin2019streaming}, which extends KLL sketch \cite{karnin2016optimal} to handle weighted items.

\section{Weighted Sampling}

GBDT learns a set of decision trees iteratively that fit residual error from the previous trees. The construction of a decision tree starts from the root node, and iteratively splits each internal node at the splitting point with largest information gain. Let $\{x_1, x_2, \ldots, x_n \}$ denote the training set, and $\{g_1, g_2, \ldots, g_n \}$ denote the gradients of each data instance in each iteration. Following \cite{ke2017lightgbm}, we use \textit{variance gain} to measure the information gain of a given candidate splitting point, which is defined as below.

\begin{definition} \label{definition:var_gain}
\cite{ke2017lightgbm} 
Let $\mathcal{O}$ be the training set on a fixed node of the decision tree.
The variance gain at splitting point $v$ for this node is defined as
$$
V(v) = \frac{1}{n} \left( \frac{ (\sum_{x_{i} \in \mathcal{O}^L(v)} g_i)^2 }{ \vert \mathcal{O}^L(v) \vert } + \frac{ (\sum_{x_{i} \in \mathcal{O}^R(v)} g_i)^2 }{ \vert \mathcal{O}^R(v) \vert } \right)
$$
where $n = \vert \mathcal{O} \vert$, $\mathcal{O}^L(v) = \{ x_{i} \in \mathcal{O} \mid x_{i} < v \}$, $\mathcal{O}^R(v) = \{ x_{i} \in \mathcal{O} \mid x_{i} > v \}$.
\end{definition}

In our proposed weighted sampling method, each data instance $x_i$ is included into the random sample $\mathcal{S}$ independently with probability $p_i = \min(1, s \vert g_i \vert / W)$, where $W = \sum_{i=1}^n \vert g_i \vert$. It is easy to see that the expected size of $\mathcal{S}$ is at most $s$.

The variance gain is estimated over $\mathcal{S}$ by
$$
\tilde{V}^{\text{WS}}(v) = \frac{1}{n} \left( \frac{ (\sum_{x_{i} \in \mathcal{S}^L(v)} g_i / p_i)^2 }{ \vert \mathcal{O}^L(v) \vert } + \frac{ (\sum_{x_{i} \in \mathcal{S}^R(v)} g_i / p_i)^2 }{ \vert \mathcal{O}^R(v) \vert } \right)
$$

\begin{lemma} \label{lemma:hoeffding-ineq}
(Hoeffding Inequality) 
Let $X_1, X_2, \ldots, X_n$ be independent random variables, such that $a_i \le X_i \le b_i$ for all $i$. Then for any $t > 0$, we have
$$
\mathsf{P}(\vert S_n - \mathsf{E}[S_n] \vert \ge t) \le 2 \exp \left(- \frac{2 t^2}{\sum_{i=1}^n (b_i - a_i)^2} \right)
$$
where $S_n = \sum_{i=1}^n X_i$.
\end{lemma}

\begin{theorem} \label{theorem:ws_err}
Let the approximation error of our weighted sampling be $\mathcal{E}^{\text{WS}}(v) = \vert \tilde{V}^{\text{WS}}(v) - V(v) \vert$. With probability at least $1 - \delta$,
$$
\mathcal{E}^{\text{WS}}(v) = O \left( \frac{W^2}{n s} \sqrt{\log (1/\delta)} \right)
$$
\end{theorem}

\begin{proof}
Let $\mathcal{Q}$ be any subset of $\mathcal{O}$. Define $w_{\mathcal{S}}(\mathcal{Q}) = \sum_{x_i \in \mathcal{\mathcal{Q}} \cap \mathcal{S}} g_i / p_i$, $w(\mathcal{Q}) = \sum_{x_i \in \mathcal{Q}} g_i$, $X_i = g_i / p_i I[x_i \in \mathcal{S}]$. It is easy to see that $w_{\mathcal{S}}(\mathcal{Q})$ is an unbiased estimator of $w(\mathcal{Q})$,
$$
\mathsf{E}[w_{\mathcal{S}}(\mathcal{Q})] = \mathsf{E} \left[  \sum_{x_i \in \mathcal{\mathcal{Q}}} X_i \right] = \sum_{x_i \in \mathcal{Q}} g_i = w(\mathcal{Q})
$$
Without loss of generality, we can assume that for each $x_i \in \mathcal{Q}$ the sampling probability $p_i < 1$ when analyzing the error $\vert w_{\mathcal{S}}(\mathcal{Q}) - w(\mathcal{Q}) \vert$.

Since the random variables $\{ X_i \}, 1 \le i \le n$ are independent and each is bounded in the range $[-W/s, W/s]$, by Hoeffding inequality,
$$
\mathsf{P}\left( \vert w_{\mathcal{S}}(\mathcal{Q}) -  w(\mathcal{Q}) \vert  \ge  \frac{2W}{s} \sqrt{\vert \mathcal{Q} \vert \log (2/\delta) } \right) \le 2 \exp \left(- \frac{\vert \mathcal{Q} \vert (2 W / s)^2 \log (2/\delta)}{\sum_{x_i \in \mathcal{Q}}(2 W / s)^2 } \right) = \delta
$$
Then with probability at least $1 - \delta$, we have
\begin{eqnarray*}
\left\vert \frac{w_{\mathcal{S}}^2(\mathcal{Q})}{\vert \mathcal{Q} \vert} -  \frac{w^2(\mathcal{Q})}{\vert \mathcal{Q} \vert} \right\vert & = & \left\vert \frac{ ( w(\mathcal{Q}) + w_{\mathcal{S}}(\mathcal{Q}) - w(\mathcal{Q}) )^2 }{\vert \mathcal{Q} \vert} -  \frac{w^2(\mathcal{Q})}{\vert \mathcal{Q} \vert} \right\vert \\
& \le & \frac{2 w(\mathcal{Q}) \cdot \vert w_{\mathcal{S}}(\mathcal{Q}) - w(\mathcal{Q}) \vert}{ \vert \mathcal{Q} \vert } + \frac{ ( w_{\mathcal{S}}(\mathcal{Q}) - w(\mathcal{Q}) )^2 }{ \vert \mathcal{Q} \vert } \\
& \le & \frac{2 W^2}{s} \sqrt{\log (2/\delta)} + \frac{W^2}{s^2} \log (2/\delta)
\end{eqnarray*}

By the inequality above and union bound, we have
\begin{eqnarray*}
\mathcal{E}^{\text{WS}}(v) & = & \frac{1}{n} \left\vert \frac{w_{\mathcal{S}}^2( \mathcal{O}^L(v) )}{\vert \mathcal{O}^L(v) \vert} -  \frac{w^2( \mathcal{O}^L(v) )}{\vert \mathcal{O}^L(v) \vert}
+ \frac{w_{\mathcal{S}}^2( \mathcal{O}^R(v) )}{\vert \mathcal{O}^R(v) \vert} -  \frac{w^2( \mathcal{O}^R(v) )}{\vert \mathcal{O}^R(v) \vert} \right\vert \\
& \le & \frac{1}{n} \left\vert \frac{w_{\mathcal{S}}^2( \mathcal{O}^L(v) )}{\vert \mathcal{O}^L(v) \vert} -  \frac{w^2( \mathcal{O}^L(v) )}{\vert \mathcal{O}^L(v) \vert} \right\vert
+ \frac{1}{n} \left\vert \frac{w_{\mathcal{S}}^2( \mathcal{O}^R(v) )}{\vert \mathcal{O}^R(v) \vert} -  \frac{w^2( \mathcal{O}^R(v) )}{\vert \mathcal{O}^R(v) \vert} \right\vert \\
& \le & \frac{4 W^2}{n s} \sqrt{\log (4/\delta)} + \frac{2 W^2}{n s^2} \log (4/\delta)
\end{eqnarray*}
\end{proof}

\textbf{Remark}. GOSS firstly sorts all the data instances according to the magnitude of their associated gradient and selects top $an$ instances, denoted as $A$, with larger gradients, then it uniformly samples a subset $B$ of size $bn$ from the remaining instances with smaller gradients, where $a$ and $b$ are the sampling ratio defined by the users. Let the approximation error of GOSS be $\mathcal{E}^{\text{GOSS}}(v) = \vert \tilde{V}^{\text{GOSS}}(v) - V(v) \vert$. The authors of \cite{ke2017lightgbm} prove that with probability at least $1 - \delta$,
$$
\mathcal{E}^{\text{GOSS}}(v) = O \left( C_{a, b}^2 \log (1/\delta) + 2 D(v) C_{a, b} \sqrt{\log (1/\delta) / n} \right)
$$
where $C_{a, b} = \frac{(1-a)}{\sqrt{b}} \max_{x_i \in A^c} \vert g_i \vert$, $D(v) = \max \left(  \frac{ \sum_{x_{i} \in \mathcal{O}^L(v)} \vert g_i \vert }{ \vert \mathcal{O}^L(v) \vert }, \frac{ \sum_{x_{i} \in \mathcal{O}^R(v)} \vert g_i \vert }{ \vert \mathcal{O}^R(v) \vert } \right)$.

Let $s = (a + b)n$ be the sample size, the upper bound of $\mathcal{E}^{\text{GOSS}}(v)$ becomes $\Omega(W^2 n / s^3 \log(1 / \delta))$ since $C_{a, b} = \Omega(\sqrt{n / s} \cdot W / s)$ and $D(v) = \Omega(W/s)$, which is asymptotically worse than our weighted sampling method. Moreover, GOSS requires selecting top-$an$ instances over the entire dataset, which will incur a large amount of communication if the dataset is distributed among multiple nodes, while our weighted sampling method only needs to synchronize the total weight $W$ of the entire dataset.

\section{Weighted Quantile Summary}

Formally, the weighted quantile problem can be described as follows. Given a dataset $\mathcal{D} = \{ (v_1, w_1), (v_2, w_2), \ldots, (v_n, w_n) \}$, where $v_1 \le v_2 \le \ldots \le v_n$, and $w_i > 0$ for $i = 1, 2, \ldots, n$. Each $v_i$ corresponds to a feature value of a data instance and $w_i$ is its weight.

Define the rank of $v$ by
$$
r_{\mathcal{D}}(v) = \sum_{i : v_i < v, (v_i, w_i) \in \mathcal{D}} w_i
$$
$$
r_{\mathcal{D}}^+(v) = \sum_{i : v_i \le v, (v_i, w_i) \in \mathcal{D}} w_i
$$
The goal is constructing a small summary to answer $\varepsilon$-approximate quantile query, \ie, to compute $\tilde{r}(v)$ for any given value $v$ such that
$$\vert \tilde{r}(v) - r_{\mathcal{D}}(v)  \vert \le \varepsilon w_{\mathcal{D}}$$
with probability $1 - \delta$, where $w_{\mathcal{D}} = \sum_{ (v_i, w_i) \in \mathcal{D} } w_i$.

\subsection{Basic Summary}

Our weighted quantile summary (bucketizer) of $\mathcal{D}$ is constructed as
$$\mathcal{T} = \{ (v_i, \tilde{w}_i)  \mid  (v_i, w_i) \in \mathcal{D}, \exists j \in \{ 0, 1, 2, \cdots, \ceil*{w_\mathcal{D}/t} \}, r_{\mathcal{D}}(v_i) \le b + j t < r_{\mathcal{D}}^+(v_i) \}$$
where $t > 0$ is a chosen parameter which we will decide later, $b$ is uniform in the range $(0, t)$, and the corresponding weight of $v_i$ is adjusted to
$$\tilde{w}_i = \left(\floor*{\frac{r^+_{\mathcal{D}}(v_i) - b}{t}} - \floor*{\frac{r_{\mathcal{D}}(v_i) - b}{t}} \right) \cdot t$$
The rank $r_{\mathcal{D}}(v)$ of any given value $v$ can be estimated using $\mathcal{T}$ by
$$
r_{\mathcal{T}}(v) = \sum_{i : v_i < v, (v_i, \tilde{w}_i) \in \mathcal{T}} \tilde{w}_i
$$

\begin{figure}
\centering
\begin{tikzpicture}[xscale=0.8]
\draw [thick] (0,0) -- (10,0);
\draw (0,0) -- (0, .1);
\draw (1,0) -- (1,.1);
\draw [fill] (1.2, 0) circle [radius=0.05];
\node[above] at (1.2, 0)%
{\scriptsize $b$};
\draw (2,0) -- (2,.1);
\draw [fill] (3.2, 0) circle [radius=0.05];
\node[above] at (3.2, 0)%
{\scriptsize $b+t$};
\draw [fill] (5.2, 0) circle [radius=0.05];
\node[above] at (5.2, 0)%
{\scriptsize $b+2t$};
\draw (5.5,0) -- (5.5,.1);
\draw (7,0) -- (7,.1);
\draw [fill] (7.2, 0) circle [radius=0.05];
\node[above] at (7.2, 0)%
{\scriptsize $b+3t$};
\draw (9,0) -- (9, .1);
\draw [fill] (9.2, 0) circle [radius=0.05];
\node[above] at (9.2, 0)%
{\scriptsize $b+4t$};
\draw (10,0) -- (10,.1);
\node[align=center, below] at (0.5,-.1)%
{$v_1$};
\node[align=center, below] at (1.5,-.1)%
{$v_2$};
\node[align=center, below] at (3.75,-.1)%
{$v_3$};
\node[align=center, below] at (6.25,-.1)%
{$v_4$};
\node[align=center, below] at (8,-.1)%
{$v_5$};
\node[align=center, below] at (9.5,-.1)%
{$v_6$};
\end{tikzpicture}

\begin{tikzpicture}[xscale=0.8]
\draw [thick] (0,0) -- (10,0);
\draw (0,0) -- (0, .1);
\draw (2,0) -- (2,.1);
\draw (6,0) -- (6,.1);
\draw (8,0) -- (8, .1);
\draw (10,0) -- (10,.1);
\node[align=center, below] at (1,-.1)%
{$v_2$};
\node[align=center, above] at (1,0)%
{\scriptsize $t$};
\node[align=center, below] at (4,-.1)%
{$v_3$};
\node[align=center, above] at (4,0)%
{\scriptsize $2t$};
\node[align=center, below] at (7,-.1)%
{$v_5$};
\node[align=center, above] at (7,0)%
{\scriptsize $t$};
\node[align=center, below] at (9,-.1)%
{$v_6$};
\node[align=center, above] at (9,0)%
{\scriptsize $t$};
\end{tikzpicture}

\caption{An example of bucketizer construction. $v_1, v_2, \ldots, v_6$ are ordered values, with their weight as the corresponding bucket's width.} \label{fig:bucketizer}
\end{figure}
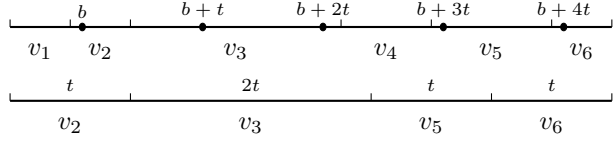

Obviously $\mathcal{T}$ contains at most $\ceil*{w_\mathcal{D}/t}$ items, and each item is represented by two records, while each item in the summary from \cite{chen2016xgboost} need four records.

\begin{lemma} \label{lemma:unbias-bucketizer}
For any given value $v$, $r_{\mathcal{T}}(v)$ is an unbiased estimator of $r_{\mathcal{D}}(v)$, with error bounded in the range $[-t, t]$.
\end{lemma}

\begin{proof}
Since the offset $b$ is uniform in the range $(0, t)$, we have
$$
r_{\mathcal{T}}(v) = \sum_{i : v_i < v, (v_i, \tilde{w}_i) \in \mathcal{T}} \tilde{w}_i =
\left\{ \begin{array}{lcl}
\floor*{ \frac{r_{\mathcal{D}}(v)}{t} } \cdot t & \mbox{w. p.} & 1 - \left\{ \frac{r_{\mathcal{D}}(v)}{t} \right\} \\ \left(\floor*{ \frac{r_{\mathcal{D}}(v)}{t} } + 1 \right) \cdot t & \mbox{w. p.} & \left\{ \frac{r_{\mathcal{D}}(v)}{t} \right\}
\end{array}\right.
$$
where $\{ x \} = x - \floor*{x}$ denotes the fractional part of $x$.

Then it is easy to show that $r_{\mathcal{T}}(v)$ is an unbiased estimator for $r_{\mathcal{D}}(v)$,
$$
\mathsf{E} [ r_{\mathcal{T}}(x) ] = \floor*{ \frac{r_{\mathcal{D}}(x)}{t} } \cdot t  +  \left\{ \frac{r_{\mathcal{D}}(x)}{t} \right\} \cdot t = r_{\mathcal{D}}(x)
$$
and the error is bounded in the range $[-t, t]$,
$$
\vert r_{\mathcal{T}}(x) - r_{\mathcal{D}}(x) \vert \le \max \left\{ \left\vert \floor*{ r_{\mathcal{D}}(v) / t } \cdot t - r_{\mathcal{D}}(x) \right\vert, \left\vert \left(\floor*{ r_{\mathcal{D}}(v) / t } + 1 \right) \cdot t - r_{\mathcal{D}}(x) \right\vert \right\} \le t
$$
\end{proof}

\subsection{Distributed Protocols}

\subsubsection{Flat model}

\textbf{The protocol}. The dataset $\mathcal{D} = \mathcal{D}_1 \cup \mathcal{D}_2 \cup \cdots \cup \mathcal{D}_k$ is distributed among $k$ nodes, on each node $j$ we construct a bucketizer $\mathcal{T}_j$ using step size $t = \varepsilon w_\mathcal{D} / \sqrt{k \log (2 / \delta)}$ from its local dataset $\mathcal{D}_j$, and then send it to the central coordinator to get the final summary  $\mathcal{T} = \cup_{j=1}^k \mathcal{T}_j$. The total communication cost is $\vert \mathcal{T} \vert$.

It is easy to see that $r_{\mathcal{T}}(v)$ is an unbiased estimator for $r_{\mathcal{D}}(v)$,
$$
\mathsf{E} [ r_{\mathcal{T}}(v) ] = \mathsf{E} \left[ \sum_{j=1}^k r_{\mathcal{T}_j}(v) \right] = \sum_{j=1}^k r_{\mathcal{D}_j}(v) = r_{\mathcal{D}}(v)
$$
Since the random variables $\{ r_{\mathcal{T}_j}(v) - r_{\mathcal{D}_j}(v) \}, 1 \le j \le k$ are independent and has bounded range $[-t, t]$, we can apply Hoeffding inequality to bound the probability that the error exceeds $\varepsilon w_\mathcal{D}$,
$$
\mathsf{P} ( \vert r_{\mathcal{T}}(v) - r_{\mathcal{D}}(v) \vert > \varepsilon w_\mathcal{D} ) \le 2 \exp \left( -\frac{2 (\varepsilon w_\mathcal{D})^2}{k t^2} \right) \le \delta
$$
Then $\vert \mathcal{T} \vert = \ceil*{w_\mathcal{D}/t} = O (\sqrt{k \log (1 / \delta)} / \varepsilon)$, compared to the summary proposed in \cite{chen2016xgboost} which has $O ( k / \varepsilon )$ communication cost.

\textbf{Reducing individual communication cost}. If the dataset $\mathcal{D}$ is not equally partitioned among these $k$ nodes, some nodes will incur more communication than others, and in the worst case all the communication is incurred from only one node by the above protocol. By varying the step size $t_j$ on each node $j = 1, 2, \cdots, k$, the maximum individual communication cost is $O(\sqrt{\log (1 / \delta)} / \varepsilon)$,
$$
t_j =
\left\{ \begin{array}{ll}
\varepsilon w_\mathcal{D} / \sqrt{k \log (2 / \delta)} & \mbox{if} \quad w_{\mathcal{D}_j} < \frac{ w_{\mathcal{D}} }{ \sqrt{k} } \\
\varepsilon w_{\mathcal{D}_j} / \sqrt{\log (2 / \delta)} & \mbox{otherwise}
\end{array}\right.
$$
The individual communication cost on each node $j$ is
$$\vert \mathcal{T}_j \vert = \ceil*{w_{\mathcal{D}_j} / t_j} = O \left( \sqrt{\log (1 / \delta)} / \varepsilon \right)$$
and the total communication cost is still $O (\sqrt{k \log (1 / \delta)} / \varepsilon)$,
$$\vert \mathcal{T} \vert  =  \sum_{j=1}^{k} \vert \mathcal{T}_j \vert  =  \sqrt{\log \frac{2}{\delta}} \left( \underset{ w_{\mathcal{D}_j} < \frac{ w_{\mathcal{D}} }{\sqrt{k}} }{ \sum }  \sqrt{k} w_{\mathcal{D}_j} / (\varepsilon w_{\mathcal{D}}) + \underset{ w_{\mathcal{D}_j} \ge \frac{ w_{\mathcal{D}} }{\sqrt{k}} }{ \sum } 1 / \varepsilon \right)  =  O \left( \sqrt{k \log \frac{1}{\delta}} / \varepsilon \right)$$

To see that the estimator $r_{\mathcal{T}}(v)$ is an $\varepsilon$-approximation of $r_{\mathcal{D}}(v)$, an elementary calculation shows that,
\begin{eqnarray*}
\sum_{j=1}^k t_j^2 \log \frac{2}{\delta}  & = &  \sum_{ w_{\mathcal{D}_j} < \frac{ w_{\mathcal{D}} }{\sqrt{k}} } \left( \frac{ \varepsilon w_{\mathcal{D}} }{\sqrt{k}} \right)^2 + \sum_{ w_{\mathcal{D}_j} \ge \frac{ w_{\mathcal{D}} }{\sqrt{k}} } \left( \varepsilon w_{\mathcal{D}_j} \right)^2 \\
& \le & k \left( \frac{ \varepsilon w_{\mathcal{D}} }{\sqrt{k}} \right)^2 + \sum_{j=1}^k \left(\varepsilon w_{\mathcal{D}_j} \right)^2 \\
& \le & k \left( \frac{ \varepsilon w_{\mathcal{D}} }{\sqrt{k}} \right)^2 + \left(\sum_{j=1}^k \varepsilon w_{\mathcal{D}_j} \right)^2 \\
& = & 2 (\varepsilon w_{\mathcal{D}})^2
\end{eqnarray*}
Then we have $\mathsf{P} ( \vert r_{\mathcal{T}}(v) - r_{\mathcal{D}}(v) \vert > \varepsilon w_\mathcal{D}) \le 2 \exp(- 2 (\varepsilon w_\mathcal{D})^2 / \sum_{j=1}^k t_j^2) \le \delta$ by Hoeffding inequality.

\subsubsection{Tree Shape All-Reduce Model}

In the tree shape all-reduce protocol, nodes with local data are organized into a binary tree structure. In the reduce phase, the data moves up the tree from leaf nodes and gets aggregated on their parent nodes. In the broadcast phase, the global aggregated result at the root node is passed down to all the other nodes. This all-reduce protocol is used by RABIT, which is the underlying all-reduce library responsible for distributed training in XGBoost \cite{chen2016xgboost}.

\begin{figure}
\centering
\begin{tikzpicture}[<-,>=stealth',level/.style={sibling distance = 5cm/#1,
  level distance = 1.5cm},cus/.style = {shape=circle, draw, align=center}] 
\node [cus] (n1) {3}
    child{ node [cus] (n2) {1} 
            child{ node [cus] (n4) {7} 
            	child{ node [cus] (n5) {5} edge from parent node[left] {5} } 
				child{ node [cus] (n6) {6} edge from parent node[left] {6} }
				edge from parent node[left] {18}
            } 
            child{ node [cus] (n16) {2}
				child{ node [cus] (n7) {7} edge from parent node[left] {7}}
				child{ node [cus] (n8) {8} edge from parent node[left] {8}}
				edge from parent node[left] {17}
            }
            edge from parent node[above left] {36}
    }
    child{ node [cus] (n3) {5}
            child{ node [cus] (n9) {4} 
				child{ node [cus] (n10) {2} edge from parent node[left] {2}}
				child{ node [cus] (n11) {1} edge from parent node[left] {1}}
				edge from parent node[left] {7}
            }
            child{ node [cus] (n12) {3}
				child{ node [cus] (n13) {9} edge from parent node[left] {9}}
				child{ node [cus] (n14) {4} edge from parent node[left] {4}}
				edge from parent node[left] {16}
            }
            edge from parent node[above left] {28}
		}
;
\draw (n1) edge[->, bend left=15] node [right] {67} (n2);
\draw (n1) edge[->, bend left=15] node [right] {67} (n3);

\draw (n2) edge[->, bend left=15] node [below] {67} (n4);
\draw (n2) edge[->, bend left=15] node [right] {67} (n16);
\draw (n4) edge[->, bend left=15] node [below] {67} (n5);
\draw (n4) edge[->, bend left=15] node [right] {67} (n6);
\draw (n16) edge[->, bend left=15] node [below] {67} (n7);
\draw (n16) edge[->, bend left=15] node [right] {67} (n8);

\draw (n3) edge[->, bend left=15] node [below] {67} (n9);
\draw (n3) edge[->, bend left=15] node [right] {67} (n12);
\draw (n9) edge[->, bend left=15] node [below] {67} (n10);
\draw (n9) edge[->, bend left=15] node [right] {67} (n11);
\draw (n12) edge[->, bend left=15] node [below] {67} (n13);
\draw (n12) edge[->, bend left=15] node [right] {67} (n14);
\end{tikzpicture}

\caption{An example of sum operation in tree shape all-reduce protocol, where each node holds a number.} \label{fig:all-reduce}
\end{figure}
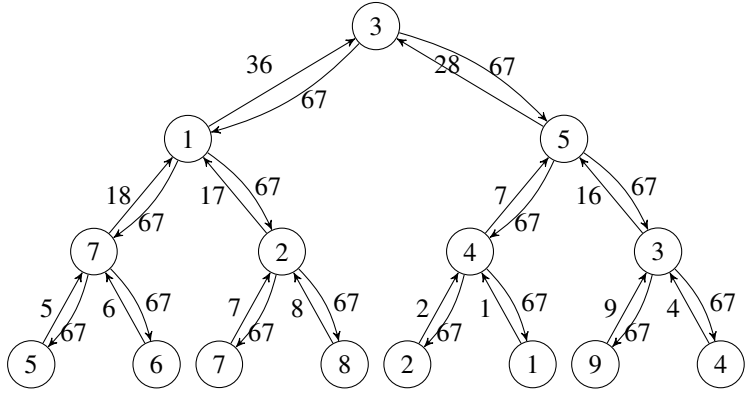

\textbf{The protocol}. \textbf{(1)} On a leaf node $j$, we construct a bucketizer $\mathcal{T}_j$ over its local dataset $\mathcal{D}_j$ using step size $t = \varepsilon w_{\mathcal{D}} / \sqrt{ 2 k \log k \log (2 / \delta) }$, and send $\mathcal{T}_j$ to its parent node. \textbf{(2)} On an internal node $j$ at level $h$, let node $j_L$ and $j_R$ denote the children of node $j$, we construct a bucketizer $\mathcal{T}_j$ with step size $t_j = (\sqrt{2})^h t$ over $\tilde{\mathcal{D}}_j = \mathcal{D}_j \cup \mathcal{T}_{j_L} \cup \mathcal{T}_{j_R}$, then send $\mathcal{T}_j$ to the parent of node $j$. \textbf{(3)} The final summary $\mathcal{T}$ at the root is an $\varepsilon$-approximate quantile summary over the entire dataset $\mathcal{D}$.

\begin{lemma} \label{lemma:azuma-hoeffding-ineq}
(Azuma-Hoeffding Inequality) 
Suppose $\{ X_i \}$, $0 \le i \le n$ is a martingale such that $X_0 = 0$ and $\vert X_i - X_{i-1} \vert \le c_i$ for $ 1 \le i \le n$. Then for any $t > 0$,
$$
\mathsf{P}(\vert X_n \vert \ge t) \le 2 \exp \left( -\frac{t^2}{2 \sum_{i=1}^n c_i^2} \right)
$$
\end{lemma}

\begin{theorem} \label{theorem:tree_quantile}
The above protocol has $O(\sqrt{k \log k \log (1 / \delta)} / \varepsilon )$ expected communication cost, and returns a summary which can answer $\varepsilon$-approximate quantile query with probability at least $1 - \delta$. Moreover, if the dataset $\mathcal{D}$ is equally partitioned among all $k$ nodes, \ie, $w_{\mathcal{D}_1} = w_{\mathcal{D}_2} = \ldots = w_{\mathcal{D}_k} = w_{\mathcal{D}}/k$, then it has $O(\sqrt{\log k \log (1 / \delta)} / \varepsilon)$ individual communication cost.
\end{theorem}

\begin{proof}
First, we analyze the error of estimating $r_{\mathcal{D}}(v)$ by $r_{\mathcal{T}}(v)$. Let $\Delta_j = r_{\mathcal{T}_j}(v) - r_{\tilde{\mathcal{D}}_j}(v) = r_{\mathcal{T}_{j}}(v) - (r_{\mathcal{D}_j}(v) + r_{\mathcal{T}_{j_L}}(v) + r_{\mathcal{T}_{j_R}}(v))$ be the error introduced at node $j$, where node $j_L$ and $j_R$ are the children of node $j$. Then $r_{\mathcal{T}}(v) - r_{\mathcal{D}}(v) = \sum_{j=1}^k \Delta_j$.

We index the nodes level by level in a bottom up fashion. The probability distribution of $\Delta_j$ depends on the value of $r_{ \mathcal{D}_j \cup \mathcal{T}_{j_L} \cup \mathcal{T}_{j_R}}(v)$, thus depends on $\{ \Delta_1, \Delta_2, \ldots, \Delta_{j-1} \}$, but the conditional expectation $\mathsf{E} [ \Delta_j \mid \Delta_i, i < j]$ is always 0. Then $X_j = \sum_{i = 1}^j \Delta_{i}$ forms a martingale with $X_0 = 0$. Since $\vert \Delta_j \vert \le t_j = (\sqrt{2})^h t$ for each node $j$ at level $h$, we can apply Azuma-Hoeffding inequality to bound the probability that the error of $r_{\mathcal{T}}(v)$ exceeds $\varepsilon w_{\mathcal{D}}$,
\begin{eqnarray*}
\mathsf{P} ( \vert r_{\mathcal{T}}(v) - r_{\mathcal{D}}(v) \vert > \varepsilon w_{\mathcal{D}} )  & \le &  2 \exp \left( -\frac{(\varepsilon w_{\mathcal{D}})^2}{2 \sum_{j = 1}^k t_{j}^2} \right) \\
& = &  2 \exp \left(-\frac{(\varepsilon w_{\mathcal{D}})^2}{2 \sum_{h = 0}^{\log k} \frac{k}{2^h} ((\sqrt{2})^h t)^2} \right) \\
& = &  2 \exp \left( -\frac{(\varepsilon w_{\mathcal{D}})^2}{2 t^2 k \log k} \right) \\
& = & \delta
\end{eqnarray*}
Next, we bound the expected communication cost. By Lemma~\ref{lemma:unbias-bucketizer}, for each node $j$ we have
$$
\mathsf{E} \left[ w_{\mathcal{T}_j} \right] = w_{\tilde{\mathcal{D}}_j}
$$
Then by law of total expectation, let $\mathcal{L}_h$ denote the set of nodes at level $h$, the total communication cost incurred by our protocol is
$$
\mathsf{E} \left[ \sum_{h = 0}^{\log k} \frac{\sum_{j \in \mathcal{L}_h} w_{\tilde{\mathcal{D}}_j} }{ (\sqrt{2})^h t} \right]
\le \sum_{h = 0}^{\log k} \frac{w_{\mathcal{D}}}{ (\sqrt{2})^h t}
= O \left(\frac{\sqrt{k \log k \log (1 / \delta)}}{\varepsilon} \right)
$$
If $w_{\mathcal{D}_1} = w_{\mathcal{D}_2} = \ldots = w_{\mathcal{D}_k} = w_{\mathcal{D}}/k$, the individual communication cost of any node $j$ is
$$
\mathsf{E} \left[ \frac{w_{\tilde{\mathcal{D}}_j}}{t_j} \right] 
= \frac{ 2^h w_{\mathcal{D}}/k }{ (\sqrt{2})^h t } 
= \frac{(\sqrt{2})^h w_{\mathcal{D}}}{k t} 
= O \left( \frac{\sqrt{\log k \log (1 / \delta)}}{\varepsilon} \right)
$$
where $h$ is the level of node $j$ in the binary tree.
\end{proof}
 
\bibliographystyle{unsrt}

\end{document}